\documentclass[a4paper,11pt]{article}

\usepackage[a4paper,margin=2.5cm]{geometry}

\usepackage[T1]{fontenc}
\usepackage[utf8]{inputenc}
\usepackage{lmodern}

\usepackage{amsmath,amssymb,amsthm}
\usepackage{bbm}
\usepackage{dsfont}
\usepackage{braket}

\usepackage{graphicx}
\usepackage{rotating}           % sidewaysfigure
\usepackage[dvipsnames]{xcolor}
\usepackage{listings}

\usepackage[numbers,sort&compress]{natbib}

\usepackage[colorlinks=true,allcolors=blue]{hyperref}

\newcommand{\numberthis}{\addtocounter{equation}{1}\tag{\theequation}} % number a single line in align*
\newcounter{Aeq}

\newcommand{\openone}{\mathds{1}}   % identity operator
\newcommand{\Hil}{\mathcal{H}}
\newcommand{\sH}{\mathcal{H}}

\newcommand{\E}{\mathcal{E}}

\newcommand{\mM}{\mathcal{M}}

\newcommand{\povm}[1]{\mathbf{#1}}
\renewcommand{\set}[1]{\mathcal{#1}} % overrides \set from braket

\newcommand{\Tr}{\operatorname{Tr}}
\newcommand{\ketbra}[1]{|{#1}\rangle\!\langle{#1}|}
\newcommand{\dket}[1]{|{#1}\rangle\!\rangle}              % double ket
\newcommand{\dbra}[1]{\langle\!\langle{#1}|}              % double bra
\renewcommand{\ge}{\geqslant}
\renewcommand{\geq}{\geqslant}
\renewcommand{\le}{\leqslant}
\renewcommand{\leq}{\leqslant}

\newtheorem{theorem}{Theorem}[section]

\newtheorem{corollary}[theorem]{Corollary}

\newtheorem{definition}[theorem]{Definition}
\theoremstyle{definition}

\theoremstyle{remark}
\newtheorem{remark}[theorem]{Remark}

\title{Quantum Measurement Retrodiction\\ and Entropic Uncertainty Relations}

\author{%
  Jiaxi Kuang\thanks{Email: \texttt{kuang.jiaxi.g5@s.mail.nagoya-u.ac.jp}},
  Kensei Torii\thanks{Email: \texttt{toriikensei@nagoya-u.jp}},
  and Francesco Buscemi\thanks{Email: \texttt{buscemi@nagoya-u.jp}}\\[1.5ex]
  {\normalsize Department of Mathematical Informatics, Nagoya University}\\
  {\normalsize Furo-cho, Chikusa-ku, Nagoya 464-8601, Japan}%
}

\date{}

\begin{document}

\maketitle

\begin{abstract}
We study quantum measurement retrodiction via the principle of minimum change. For general quantum measurements that are described by quantum-to-classical channels, we show that a broad family of standard quantum divergences admits a common optimal retrodictive update, yielding a canonical and divergence-independent quantum Bayesian inverse for any POVM and prior state. Using this update, we construct a symmetric joint distribution for pairs of POVMs and introduce the \emph{mutual retrodictability}, which quantifies how well the two POVMs can retrodict each other's outcome distributions under this update. We also derive a general upper bound on this quantity, which depends only on the prior and holds for all measurements. This framework leads to two \emph{retrodictive entropic uncertainty relations}, expressed directly in terms of the prior and the POVMs, yet valid independently of any retrodictive interpretation and fully compatible with conventional operational formulations. One of these relations links entropic uncertainty, approximate recoverability, and the thermodynamics of measurement through the Groenewold--Ozawa information gain. Finally, numerical benchmarks show that our second bound is consistently tighter than existing entropic uncertainty relations across broad classes of measurements and states.
\end{abstract}

\section{Introduction}

Let us begin by considering a finite-state classical system, whose states are labeled by the elements of a finite set $\set{S}=\{s\}$. Suppose that the actual, or ``true,'' state of the system is unknown to us. Instead, our uncertainty about it is captured by a prior distribution $\gamma(s)$. The system then undergoes a measurement whose possible outcomes form another finite set $\set{X}=\{x\}$, and whose statistics are governed by a known likelihood function $\varphi(x|s)$, the probability of observing outcome $x$ given that the system is in state $s$.

When a particular outcome $\bar x\in\set{X}$ is observed, Bayes' rule prescribes how our prior beliefs should be updated in light of this new information:
\begin{align}\label{eq:vanilla-bayes}
\hat\gamma(s|\bar x):=\frac{\gamma(s)\;\varphi(\bar x|s)}{\sum_s\gamma(s)\;\varphi(\bar x|s)}\;.
\end{align}
This simple expression lies at the heart of all probabilistic inference and numerous rigorous derivations have been proposed to justify its universal applicability and remarkable empirical success across the sciences~\cite{polya1954induction,cox1961algebra,definetti1974theory,jeffrey,zellner1988optimal,jeffreys1998theory,jaynes_2003}.

Once the prior has been updated to incorporate the newly acquired information, namely the observed outcome $\bar{x}$, it can be used to compute the statistics associated with other observations. Consider, for instance, another measurement with outcomes in the set $\set{Y}=\{y\}$ and likelihood function $\psi(y|s)$. The updated belief $\hat{\gamma}(s|\bar{x})$ allows us to compute the expected probability of each outcome $y$ as
\begin{equation}\label{eq:classical-retrodiction}
\Pr\{y|\bar{x}\}=\sum_{s\in\set{S}}\hat{\gamma}(s|\bar{x})\psi(y|s)\,.
\end{equation}
This quantity $\Pr\{y|\bar{x}\}$ represents the probability assigned to outcome $y$, conditioned on the information provided by the observation $\bar{x}$. At this stage, the role of the variable $y$ is \textit{inferential} rather than \textit{temporal}: it labels the outcomes of another observation whose statistics are evaluated conditional on $\bar{x}$, without yet assigning that observation a specific position in time.

The above argument, a direct consequence of Bayes' rule and the total law of probability, admits a natural generalization to quantum theory. Consider a bipartite quantum system described by a state $\rho_{AB}$, and suppose we perform a local measurement on subsystem $A$ described by a positive operator-valued measure (POVM) $\povm{M}:=\{M_x:x\in\set{X}\}$, i.e., a family of positive semidefinite operators $M_x\ge 0$ such that $\sum_{x\in\set{X}}M_x=\openone$. Upon obtaining a specific outcome $\bar x\in\set{X}$, the state of subsystem $B$ is updated, according to Born’s rule, to
\[
\rho_{B|\bar x}:=\frac{\Tr_A\{(M_{\bar x}\otimes\openone_B)\ \rho_{AB}\}}{\Tr\{(M_{\bar x}\otimes\openone_B)\ \rho_{AB}\}}\;.
\]
This updated state encodes our new expectations about any subsequent measurement that might be performed on $B$.

In the classical case, the quantity $\hat{\gamma}(s|\bar{x})$ also admits a \textit{retrodictive interpretation}, namely as an inference about the system’s state \textit{prior to the measurement}. Retrodictive reasoning underlies the logic of diagnostic sciences, from medicine to forensics. Though the explicit term ``retrodiction'' emerged only in the past century, its conceptual roots trace back to Laplace, who described it as ``the probability of causes given events''~\cite{laplace-transl}. 

In quantum mechanics, however, this \textit{backward-in-time interpretation} encounters deep conceptual subtleties: measurement outcomes do not, in general, reveal preexisting properties of the system. Consequently, the retrodictive use of Bayes’ rule in quantum theory has long been a fertile ground for discussion and insight; for a recent overview of these developments, see Refs.~\cite{barnett-2021-retro-review,Jeffers2024QuantumRetrodiction}. A possible source of disagreement in the literature may lie in the variety of approaches that have been adopted. To the best of our knowledge, the quantum versions of classical Bayesian retrodiction proposed so far typically fall into one of the following categories: they are either \textit{semiclassical}, where the underlying process is quantum but the retrodiction concerns only the classical statistics of the outcomes (such as in Refs.~\cite{watanabe55,barnett-pegg-jeffers}); or they are defined \textit{by analogy}, in the sense that the resulting expressions reduce to classical retrodiction when operators commute, yet may yield non-positive quasi-probability distributions in the general case (such as in~\cite{johansen-2007-weak-meas-and-quasi-prob}); or they are derived from a set of \textit{axioms} or desiderata whose appeal and physical justification can vary~\cite{Parzygnat2023axiomsretrodiction}.

Only very recently, a fully quantum generalization of Bayes' rule has been derived from the \textit{minimum change principle}~\cite{bai-2025-q-bayes-minimum-change}. This principle postulates that the update of one's belief should proceed through a reverse (i.e., retrodictive) process that remains consistent with the newly acquired information while, at the same time, minimizing the statistical divergence from the corresponding forward (i.e., predictive) process~\cite{zellner1988optimal}. In essence, it embodies a conservative epistemic stance, ensuring that the belief revision faithfully reflects the new data without introducing unintended or spurious biases. In this respect, it seems to be related (at least ``morally'') with Jaynes' Maximum Entropy Principle~\cite{jaynes1957information,jaynes_2003,Caticha_2006}, although see Refs.~\cite{Seidenfeld_1986} and~\cite{UFFINK1996}.

In this paper, we focus on the task of retrodiction based on the outcome of a quantum measurement. We show that, by restricting the analysis to quantum-to-classical (measurement) channels, the quantum minimum change principle of Ref.~\cite{bai-2025-q-bayes-minimum-change} naturally extends to a broad family of statistical divergences. All such divergences yield the same optimal retrodicted quantum state, revealing the universality of the retrodictive update rule within this semiclassical setting. This universality shows that the resulting update rule is not an artifact of a particular divergence measure, but rather a robust feature of the underlying variational structure of quantum inference. Building upon this result, we employ the retrodicted state to formulate a new class of \textit{retrodictive entropic uncertainty relations}, which capture the information-theoretic trade-offs inherent in backward-in-time quantum reasoning. In addition, we construct two new entropic uncertainty relations. In particular, Theorem~\ref{th:REUR-2} shows that one of our entropic uncertainty relations admits an extension that is valid for arbitrary states and POVMs and that can be written directly in terms of the Groenewold--Ozawa information gain, thereby providing an explicit bridge to thermodynamic analyses of quantum measurement processes~\cite{buscemi2008global,buscemi-das-wilde-2016approximate,danageozian2021thermodynamic,Minagawa2025SecondLaw}. Importantly, although derived through a retrodictive construction, these relations are valid independently of that narrative and admit the conventional operational interpretation of entropic uncertainty relations. Finally, we benchmark these relations against the well-known entropic uncertainty bounds of Ref.~\cite{berta2010uncertainty}, showing that our approach yields tighter bounds in a wide range of scenarios.

The paper is organized as follows. In Section~\ref{sec:2} we recall the minimum change principle and specialize it to quantum measurement channels, showing that this restriction allows us to extend the principle from fidelity to a broad class of statistical divergences and that all such divergences select the same minimum--change retrodictive update. In Section~\ref{sec:3} we examine the resulting retrodicted states and develop the symmetric retrodictive joint probability that underpins our later constructions. In Section~\ref{sec:4} we introduce the mutual retrodictability $R(\povm{M};\povm{N})_\gamma$ and establish its general upper bound in Theorem~\ref{th:mutual-bound-2}. In Section~\ref{sec:retro-EUR} we use this framework to derive the retrodictive entropic uncertainty relations \eqref{eq:REUR-1} and~\eqref{eq:REUR-2}. In Section~\ref{sec:numericalresult} we benchmark our bounds numerically against the entropic uncertainty relation of Berta \textit{et al.}~\cite{berta2010uncertainty}. Finally, in Section~\ref{sec:conclusion} we summarize our results and outline directions for future work.

\section{Framework of retrodiction from quantum measurements}

\subsection{Preliminaries and problem setup}
For completeness, we briefly recall how the minimum change principle provides a general foundation for belief updating in a classical information-theoretic setting. We view $\varphi(x|s)$ as a classical channel from $S$ to $X$, and $\gamma(s)$ as the input distribution. The joint law
\[
P_{\rm fwd}(s,x):=\varphi(x|s)\gamma(s)\;,
\]
therefore represents the complete operational process describing the information flow from system states to observable outcomes. Here $\gamma(s)$ encodes the prior belief about $S$, while the likelihood $\varphi(x|s)$ specifies the channel correlation between $S$ and $X$.

When new information about $X$ is acquired, it typically takes the form of a constraint on the marginal of $X$. We denote this target distribution by $\tau(x)$. It may represent either a fully specified measurement result, as in the textbook Bayes update where $\tau(x)=\delta(x,\bar x)$, or a noisy or coarse-grained form of evidence~\cite{jeffrey}. The updated belief should incorporate this constraint while deviating from the prior operational process as little as possible.

The requirement of minimum information change can be cast as an information projection problem:
\begin{align} \label{eq:clax_optimize}
    \min_{R} \mathbb{D}(R\|P_{\rm fwd})\;,
\end{align}
where $\mathbb{D}$ is a statistical divergence. The variable $R$ ranges over all joint distributions consistent with the acquired information, i.e., of the form $R(s,x)=\tau(x)\hat\varphi(s|x)$ for some normalized conditional distribution $\hat\varphi(s|x)$. Thus the optimization in~\eqref{eq:clax_optimize} is over the posterior channel $\hat\varphi(s|x)$ only. In the language of information geometry, this corresponds to the projection of $P_{\rm fwd}$ onto the affine set $\{R_{SX}:R_{X}=\tau\}$, being $R_X$ the marginal of $R_{SX}$.

A key feature of the minimum change principle is that it applies to the full joint input--output process, rather than only to its marginals. In accordance with Bayesian practice, we assume that the prior process assigns nonzero probability to every trajectory, i.e., $P_{\rm fwd}(s,x)>0$ for all $s\in\set{S}$ and $x\in\set{X}$. Under this regularity assumption, and for a broad class of divergences (including all strictly convex $f$-divergences), the minimizer is unique and given by
\begin{align}
P_{\rm rev}:=\arg\min_{R} \mathbb{D}(R\|P_{\rm fwd})=\frac{P_{\rm fwd}(s,x)}{\sum_{s'} P_{\rm fwd}(s',x)} \tau(x)\;,
\end{align}
thus recovering the classical Bayesian update.

In this paper we focus on a generalization of the minimum change principle to the special case of quantum measurements. We consider a finite-dimensional quantum system $Q$ with Hilbert space $\sH_Q$, subjected to a measurement represented by a positive operator-valued measure (POVM) $\povm{M}=\{M_x:x\in\set{X}\}$, where $\set{X}$ is a finite set of outcomes and $M_x$ are positive semidefinite operators satisfying the completeness relation $\sum_{x\in\set{X}}M_x=\openone_Q$. We say that the POVM $\povm{M}$ is \textit{rank-one} whenever all its non-zero elements are rank-one operators. If the state of the system is represented by a density operator $\gamma_Q$, the probability of obtaining the outcome $x$ is given by $p(x)=\Tr\{M_x\gamma_Q\}$. It is convenient to summarize the resulting outcome distribution by introducing a quantum-to-classical channel, namely a completely positive trace-preserving map
\begin{align}\label{eq:qc-channel}
\mM(\gamma_Q):=\sum_{x\in\set{X}}\Tr\{M_x\gamma_Q\}\ketbra{x}_X\;.
\end{align}
Here the kets $\{\ket{x}\}_x$ form a fixed orthonormal basis of a space $\sH_X\cong\mathbb{C}^{|\set{X}|}$. This is a purely formal device: we will never consider superpositions of such vectors. It provides a compact way to treat the classical random variable $X$ representing the measurement outcome within the quantum formalism.

Finally, we introduce three basic information-theoretic quantities. For a classical random variable $X$ with probability distribution $\Pr\{x\}$, the Shannon entropy is defined as
\begin{equation*}
    H(X):= -\sum_x \Pr\{x\}\log_2 \Pr\{x\}\,.
\end{equation*}
Given two classical random variables $X$ and $Y$ with joint distribution $\Pr\{x,y\}$, the classical mutual information is defined as
\begin{equation*}
    I(X;Y):= H(X)+H(Y)-H(XY)\,,
\end{equation*}
where $H(Y)$ and $H(XY)$ denote the Shannon entropies of the marginal distribution $\Pr\{y\}$ and the joint distribution $\Pr\{x,y\}$, respectively. For the quantum state $\gamma_Q$, the von Neumann entropy is defined as
\begin{equation*}
    H(\gamma_Q) := -\Tr\{\gamma_Q\log_2\gamma_Q\}\,.
\end{equation*}

\subsection{Minimum change principle for quantum measurements}\label{sec:2}

In Ref.~\cite{bai-2025-q-bayes-minimum-change}, a framework was developed for studying the minimum change principle for general quantum channels. In this work, we restrict attention to quantum-to-classical (qc) channels of the form~\eqref{eq:qc-channel}. Their semiclassical structure makes them considerably easier to analyze, and allows us to extend the minimum change principle to a broad family of statistical divergences, unlike Ref.~\cite{bai-2025-q-bayes-minimum-change} which focused on fidelity.

Considering only qc channels also provides a natural analogue of the joint input--output distributions used in the classical setting. In fully noncommutative scenarios, the absence of a straightforward quantum analogue of an input–output joint state complicates the formulation~\cite{bai-2025-q-bayes-minimum-change}. For measurement channels, however, we may introduce the bipartite state
\begin{align}\label{eq:forward-cq-state}
\gamma_{XQ}:=\sum_{x\in\set{X}}\ketbra{x}_X\otimes\sqrt{\gamma_Q}M_x\sqrt{\gamma_Q}\;.
\end{align}
Up to a partial transposition, which is irrelevant here because $\gamma_{XQ}$ is separable, this forward-process bipartite state coincides with the definition in~\cite{bai-2025-q-bayes-minimum-change}. That construction, in turn, is based on the Choi isomorphism and the standard purification, both of which are commonly regarded as \textit{canonical} constructions~\cite{luo-channel-state-duality-2013,paulsen-shultz-2013-choi-variants,wilde2017quantum-book,kye-2022-choi-revisited,frembs-cavalcanti-2024-variations-choi}. Although the derivation is straightforward, for the sake of completeness, we report the explicit calculation leading to~\eqref{eq:forward-cq-state} in Appendix~\ref{app:gammaXQ}. It is immediate to verify that $\Tr_Q\{\gamma_{XQ}\}=\sum_xp(x)\ketbra{x}_X$ and $\Tr_X\{\gamma_{XQ}\}=\gamma_Q$, so the marginals recover the expected output distribution and the input state, respectively.

The minimum change principle for quantum measurement channels can now be formulated as an information projection
\begin{align}\label{eq:q_optimize}
\min_{\sigma_{XQ}} \mathbb{D}(\sigma_{XQ}\|\gamma_{XQ})\;,    
\end{align}
between the forward bipartite state $\gamma_{XQ}$, which encodes the prior process for the measurement of $Q$, and a backward bipartite state $\sigma_{XQ}$ satisfying the following conditions:
\begin{enumerate}
    \item its marginal state satisfies $\sigma_X:=\Tr_Q\{\sigma_{XQ}\}=\sum_{x\in\set{X}}\tau(x)\ketbra{x}_X$, representing the newly acquired information about the outcome distribution -- this condition is imposed as a constraint in the minimization;
    \item the backward channel, the analogue of $\hat\varphi(s|x)$ in~\eqref{eq:clax_optimize}, is now a \textit{classical-to-quantum} (cq) channel that associates to each $x\in\set{X}$ a quantum state $\sigma_{Q|x}$ on $\sH_Q$.
\end{enumerate}
Consequently, the bipartite state $\sigma_{XQ}$ takes the form
    \[
    \sigma_{XQ}= \sum_{x\in\set{X}}\tau(x)\ketbra{x}_X\otimes\sigma_{Q|x}\;.
    \]
With the above notations and assumptions, we have the following:

\begin{theorem}\label{th:divergences}
    Assume $\Tr\{\gamma_QM_x\}>0$ for all $x\in\set{X}$, and let an arbitrary probability distribution $\tau(x)$ represent the newly acquired information. For the following quantum divergences (below, $A$ and $B$ denote density matrices, and for simplicity we assume $B>0$, namely, $\gamma_{XQ}>0$):
    \begin{enumerate}
        \item Umegaki relative entropy~\cite{umegaki-q-rel-ent-1961}: $\mathbb{D}(A\|B)\equiv D(A\| B):=\Tr\{A(\log_2 A-\log_2 B)\}$;
        \item Petz--R\'enyi relative entropy ~\cite{petz1985quasi,petz1986quasi}: $D_\alpha(A\|B):=\frac{1}{\alpha-1}\log_2\Tr\{A^\alpha B^{1-\alpha}\}$, for $\alpha\in(0,1)\cup(1,2]$;
        \item sandwiched R\'enyi relative entropy ~\cite{muller2013quantum,wilde2014strong}:\\ $\tilde{D}_\alpha(A\|B):=\frac{1}{\alpha-1}\log_2\Tr\bigl\{\bigl(B^{\frac{1-\alpha}{2\alpha}}AB^{\frac{1-\alpha}{2\alpha}}\bigr)^\alpha\bigr\}$, for $\alpha\in[1/2,1)\cup(1,\infty)$;
        \item geometric R\'enyi relative entropy ~\cite{petz1998constraction,matsumoto2015new}:\\ $\hat{D}_\alpha(A\|B):=\frac{1}{\alpha-1}\log_2\Tr\{B\left(B^{-1/2}AB^{-1/2}\right)^\alpha\}$, for $\alpha\in(0,1)\cup(1,2]$;
        \item Belavkin--Staszewski relative entropy ~\cite{belavkin1982c}: $D_{BS}(A\|B):=\Tr\{A\log_2\left(A^{1/2}B^{-1}A^{1/2}\right)\}$;
    \end{enumerate}
    we have
    \[
        \min_{\sigma_{XQ}} \mathbb{D}(\sigma_{XQ}\|\gamma_{XQ})=\mathbb{D}(\sigma_{X}\|\gamma_{X})\equiv\mathbb{D}(\tau\|p)\;,
    \]
    where $\tau\equiv \tau(x)$ and $p\equiv p(x)=\Tr\{\gamma_Q M_x\}$, and the minimum is uniquely achieved by
    \begin{align}\label{eq:optimal-retrodiction}
    \hat\gamma_{XQ}&:=\arg\min_{\sigma_{XQ}} \mathbb{D}(\sigma_{XQ}\|\gamma_{XQ})=\sum_{x\in\set{X}}\tau(x)\ketbra{x}_X\otimes\frac{\sqrt{\gamma_Q}M_x\sqrt{\gamma_Q}}{\Tr\{M_x\gamma_Q\}}\;.
    \end{align}
\end{theorem}

 Besides relative entropies, the same conclusion holds also when the statistical divergence between forward and reverse processes is measured by Uhlmann fidelity (this follows directly from Ref.~\cite{bai-2025-q-bayes-minimum-change}), and the state is also optimal when using the trace-distance:

\begin{theorem}\label{th:trace-dist}
    With the same conventions as above, the minimum
    \[
    \min_{\sigma_{XQ}} \frac12\left\|\sigma_{XQ}-\gamma_{XQ}\right\|_1=\frac 12\sum_{x\in\set{X}}|\tau(x)-p(x)|\;,
    \]
    is achieved by $\hat\gamma_{XQ}$ in~\eqref{eq:optimal-retrodiction}.
\end{theorem}
The proofs of Theorems~\ref{th:divergences} and~\ref{th:trace-dist} can be found in Appendix~\ref{app:proofs-secII}.

Theorems~\ref{th:divergences} and~\ref{th:trace-dist} reinforce the argument first proposed in~\cite{bai-2025-q-bayes-minimum-change}, which motivates the adoption of the conditional states
\begin{align}\label{eq:retro-states}
    \hat\gamma_{Q|\bar x}:=\frac{\sqrt{\gamma_Q}M_{\bar x}\sqrt{\gamma_Q}}{\Tr\{\gamma_Q M_{\bar x}\}}\;,
\end{align}
as \textit{the} quantum analogue of the classical Bayesian inverse in Eq.~\eqref{eq:vanilla-bayes}. This identification had been anticipated in earlier work, where expressions equivalent to $\hat\gamma_{Q|\bar x}$ were motivated primarily by formal analogy with classical Bayes’ rule rather than derived from a variational principle; see, e.g., Refs.~\cite{HUGHSTON199314,MichaelJWHall1997,Fuchs-2003-quantum-theory-info-something-more,leifer-2006-q-dyn-as-cond-prob,leifer2007conditional,leifer-spekkens-2013-causally-neutral,jacobs-changing-mind,Parzygnat2023axiomsretrodiction,buscemi2022observational,nagasawa-2025-general-increase}. In light of Theorems~\ref{th:divergences} and~\ref{th:trace-dist}, we adopt $\hat\gamma_{Q|\bar x}$ as the canonical update singled out by the minimum-change principle, and we refer to it as the \textit{quantum Bayesian inverse} associated with outcome $\bar x$. This update has already proven useful beyond the present setting: Ref.~\cite{Liu2025retrodictive} showed that it leads naturally to a quantum analogue of classical retrodictive smoothing, and Ref.~\cite{liu2025unifying} further developed a generalized smoothing framework that treats smoothing as intrinsically retrodictive and is built explicitly on $\hat\gamma_{Q|\bar x}$. Note also that, within this semiclassical measurement framework, the construction reduces to the classical minimum-change principle in the commuting limit. In what follows, we use the Bayesian inverse $\hat\gamma_{Q|\bar x}$ to quantify how a measurement outcome can be employed to retrodict information about the system’s pre-measurement state.

Before concluding this section, we note that the states defined in~\eqref{eq:retro-states} coincide with the outputs of the \textit{Petz transpose map}~\cite{petz1986sufficient,petz1988sufficiency} associated with the measurement channel $\mM$ in~\eqref{eq:qc-channel} and the prior state $\gamma_Q$, when this map is applied to the classical output states $\ketbra{x}$. For a detailed discussion of this construction, its mathematical and statistical properties, and its physical interpretation, we refer the reader to Ref.~\cite{Nagasawa-2025-macrostates-ROPP}.

\subsection{Quantum Bayesian retrodiction from a measurement's outcomes}\label{sec:3}

We start with a few words of warning. The use of retrodictive arguments in quantum theory is conceptually delicate. We are aware that interpreting quantum measurements as providing information about prior states, in particular the interpretation of the retrodicted states $\hat\gamma_{Q|x}$, is problematic both operationally and philosophically. Nevertheless, we take a pragmatic view: we adopt the principle of minimum change as a working hypothesis and examine where it leads. Our goal is not to justify retrodictive reasoning but to study its mathematical structure and to identify the regimes where it becomes consistent.

We repeat the setup. A quantum system $Q$ has a state described by a density operator $\gamma_Q$, which, in accordance to the Bayesian narrative, we interpret as the prior state. The system undergoes a measurement represented by a POVM $\povm{M}=\{M_x:x\in\set{X}\}$, yielding an outcome $\bar x$. We now use the prior state together with the new information provided by the observed outcome $\bar{x}$ to compute the corresponding retrodicted state and, from it, the probabilities associated with the outcomes of \textit{another measurement}. Here it is important to distinguish inferential order from physical time order. Inferentially, the logic is fixed: one starts from the prior state, conditions on the observation of $\bar{x}$, forms the retrodicted state, and then evaluates the statistics of another POVM on that retrodicted state. Physically, however, this second measurement need not be actually performed, nor is it assigned, in the general discussion, a fixed temporal position relative to the observation of $\bar{x}$, although in Section~\ref{sec:retro-EUR} we show that the same joint statistics can be realized through a suitable sequential construction. In this sense, the role of the ``other measurement'' is to test how much information the outcome $\bar{x}$ conveys about the system, under the retrodictive update, rather than to represent a second measurement event placed at a definite later time.

If the other measurement is represented by a POVM $\povm{N}:=\{N_y:y\in\set{Y}\}$, the analogue of Eq.~\eqref{eq:classical-retrodiction} takes the form
\begin{align*}
\Pr\{y|\bar x\}&=\Tr\{N_y \hat\gamma_{Q|\bar x}\}\\
&=\frac{1}{\Tr\{M_{\bar x} \gamma_Q\}}\Tr\{N_y \sqrt{\gamma_Q} M_{\bar x} \sqrt{\gamma_Q}\}\;.
\end{align*}
Representing this as a diagram:
\[
    \gamma_Q\ \longrightarrow\  \text{observe } \bar x\  \longrightarrow\ \hat\gamma_{Q|\bar x}\ \longrightarrow\ \text{compute }\Pr\{y|\bar x\}=\operatorname{Tr}\{N_y\;\hat\gamma_{Q|\bar x}\}\;.
\]

Accordingly, we define the \textit{joint} retrodictive probability of outcomes $y$ of $\povm{N}$ and $x$ of $\povm{M}$ as
\begin{align}\label{eq:retro-joint-pd}
\Pr\{y\leftarrow x\}:=\Tr\{N_y \sqrt{\gamma_Q} M_x \sqrt{\gamma_Q}\}\;.
\end{align}
It is immediate to verify that marginalizing the joint probability $\Pr\{y\leftarrow x\}$ recovers the standard Born rule for both POVMs:
\begin{align}\label{eq:marginals-OK}
\sum_{y\in\set{Y}}\Pr\{y\leftarrow x\}&=\Tr\{M_x \gamma_Q\}\;, \nonumber\\
\sum_{x\in\set{X}}\Pr\{y\leftarrow x\}&=\Tr\{N_y \gamma_Q\}\;.
\end{align}
The retrodictive joint probability $\Pr\{y\leftarrow x\}$ thus defines a valid joint distribution for any state $\gamma_Q$ and any pair of POVMs $\povm{M}$ and $\povm{N}$, even when these are not jointly measurable~\cite{Buscemi2023-instrument-incompatibility}. This is possible because, in general, $\Pr\{y\leftarrow x\}$ is not a linear functional of the state $\gamma_Q$, although it remains linear in the POVMs.

We also observe that, as a consequence of the cyclic property of the trace, the retrodictive joint probability $\Pr\{y\leftarrow x\}$ is \textit{symmetric} with respect to the two POVMs, so that, in fact
\[
\Pr\{y\leftarrow x\}=\Pr\{x\leftarrow y\}\;.
\]
In other words, performing the measurement $\povm{M}$ first and using its outcome to retrodict the statistics of $\povm{N}$ yields the same joint distribution as performing $\povm{N}$ first and using its outcome to retrodict those of $\povm{M}$. This symmetry contrasts sharply with the forward-in-time, or predictive, case, where the joint statistics depend on the specific order of the measurements. In the predictive setting, the measurement of a POVM $\povm{M}=\{M_x:x\in\set{X}\}$ must be described by a \textit{quantum instrument} $\{\mathcal{J}_x^{\povm{M}}:x\in\set{X}\}$, i.e., a collection of completely positive trace-nonincreasing maps satisfying $\Tr\{\mathcal{J}_x^{\povm{M}}(\rho)\}=\Tr\{M_x \rho\}$ for all density operators $\rho$. The joint probability for two sequential measurements associated with the POVMs $\povm{M}$ and $\povm{N}$, performed in this order, is then given by
\[
\Pr\{x\to y\} = \Tr\{N_y\,\mathcal{J}_x^{\povm{M}}(\gamma_Q)\}\;,
\]
which in general depends on the particular instrument chosen for the measurement performed first.

\section{Results}
\subsection{Mutual retrodictability}\label{sec:4}

With a joint probability at hand, it is natural to define a \textit{mutual retrodictability} as follows:

\begin{definition} \label{def:mutual_retrodictability}
    Given a state $\gamma_Q$ and two finite POVMs $\povm{M}=\{M_x:x\in\set{X}\}$ and $\povm{N}=\{N_y:y\in\set{Y}\}$ on $Q$, the {\em mutual retrodictability} of $\povm{M}$ and $\povm{N}$ with respect to $\gamma_Q$ is defined as
    \[
    R(\povm{M};\povm{N})_\gamma:=I(X;Y)\;,
    \]
    where the right-hand side is the mutual information computed for the joint probability distribution $\Pr\{y\leftarrow x\}$ given in~\eqref{eq:retro-joint-pd}.
\end{definition}

The term ``mutual retrodictability'' highlights that the relationship between two measurements, $\povm{M}$ and $\povm{N}$, is not merely statistical but inferential. Through the principle of minimum change, the outcome of one measurement provides the most conservative retrodictive estimate of what the other measurement's outcome statistics must have been, had it probed the same underlying preparation. Accordingly, the mutual information of the retrodictive joint distribution $\Pr\{y \leftarrow x\}$ quantifies how strongly these backward inferences are correlated, measuring the degree to which the two POVMs can infer each other's outcome distributions under the minimum-change retrodictive update.

\begin{theorem}\label{th:mutual-bound-2}
    For any state $\gamma_Q$ and any pair of POVMs $\povm{M} = \{M_x : x \in \set{X}\}$ and $\povm{N} = \{N_y : y \in \set{Y}\}$ on $Q$, the mutual retrodictability satisfies
    \begin{align}\label{eq:mutual-bound}
        0 \le R(\povm{M}; \povm{N})_\gamma \le H(\gamma_Q)\;.
    \end{align}
\end{theorem}
The proof is reported in Appendix~\ref{app:proofIII2}.

\begin{remark}
    The theorem shows that the mutual retrodictability of any two measurements is upper-bounded by a term that depends only on the ``purity'' of the prior state. In particular, if the prior state $\gamma_Q$ is pure, then $H(\gamma_Q) = 0$ and the mutual retrodictability vanishes identically for all pairs of POVMs. In this case, the system is already in a state of complete certainty, leaving no room for non-trivial retrodictive inference. A pure prior state thus acts as an \textit{inferential firewall}: it blocks any attempt to draw further retrodictive conclusions, just as, in the classical setting, no amount of new information can alter a delta-distributed prior.
\end{remark}

\

\begin{remark}\sloppy
The bounds in~\eqref{eq:mutual-bound} are tight.
It is straightforward to construct scenarios where $R(\povm{M}; \povm{N})_\gamma = 0$ or $R(\povm{M}; \povm{N})_\gamma = H(\gamma_Q)$. The former occurs, as noted above, when, e.g., $\gamma_Q$ is a pure state; in this case, the upper bound also vanishes, so the lower and upper bounds coincide. A slightly less trivial example arises when, for instance, $\povm{M}$ is a projective measurement on the eigenstates of $\gamma_Q$, while $\povm{N}$ is mutually unbiased. Conversely, a situation where $R(\povm{M}; \povm{N})_\gamma = H(\gamma_Q) > 0$ may occur when $\povm{M} = \povm{N}$, and both are projective measurements on the eigenstates of $\gamma_Q$.
\end{remark}

\

These observations suggest the possibility that mutual retrodictability could serve as a \textit{state-dependent indicator of compatibility} between two POVMs. When $R(\povm{M};\povm{N})_\gamma$ is large, the measurements exhibit significant inferential overlap with respect to $\gamma_Q$; when it is small, they appear to probe largely independent aspects of the state. We leave this point open for future investigation.

These examples also suggest that large mutual retrodictability corresponds to regimes in which the overall triplet, consisting of the state $\gamma_Q$ together with the two POVMs $\povm{M}$ and $\povm{N}$, behaves in a quasi-classical manner. In such cases, the measurement statistics admit an interpretation in terms of classical uncertainty about an underlying ontic configuration, rather than genuine quantum indeterminacy. Put differently, when $R(\povm{M};\povm{N})_\gamma$ approaches its upper bound, the retrodictive reconstruction of measurement outcomes becomes nearly indistinguishable from standard Bayesian inference over classical variables. The notion of ``retrodiction'' thus regains its uncontroversial classical meaning: it no longer reflects an inference about noncommuting observables, but merely the update of classical ignorance concerning a common underlying cause consistent with both measurements.

\subsection{Retrodictive entropic uncertainty relations}\label{sec:retro-EUR}

Before proceeding, we emphasize that the results in this section do not rely on any interpretive use of retrodiction. The quantity $\Pr\{y \leftarrow x\}$ introduced above can be regarded simply as a symmetric joint probability constructed from two POVMs and a state, without invoking any notion of temporal order or inference about past events. The uncertainty relations that follow are purely mathematical consequences of this symmetric construction. They hold regardless of the interpretation given to $\Pr\{y \leftarrow x\}$.

As shown in Eq.~\eqref{eq:marginals-OK}, the marginals of the symmetric joint probability $\Pr\{y \leftarrow x\}$ coincide with the standard outcome probabilities of both POVMs. This allows us to identify $H(X)$ and $H(Y)$, computed from $\Pr\{y \leftarrow x\}$, with $H(\povm{M})_\gamma$ and $H(\povm{N})_\gamma$, respectively. Since the mutual retrodictability is non-negative, we immediately obtain the lower bound
\[
H(\povm{M})_\gamma + H(\povm{N})_\gamma \ge H(XY)\;,
\]
where $H(XY)$ is the entropy computed from $\Pr\{y \leftarrow x\}$. This yields our first retrodictive entropic uncertainty relation:

\begin{theorem}
    For any state $\gamma_Q$ and any pair of POVMs $\povm{M} = \{M_x : x \in \set{X}\}$ and $\povm{N} = \{N_y : y \in \set{Y}\}$ on $Q$, the following bound holds:
    \begin{align}\label{eq:REUR-1}
        H(\povm{M})_\gamma + H(\povm{N})_\gamma \ge -\max_{x,y} \log_2 \|\sqrt{N_y}\sqrt{\gamma_Q}\sqrt{M_x}\|^2_{2}\;.
    \end{align}
\end{theorem}

\begin{proof}
    The proof is immediate: the joint probability~\eqref{eq:retro-joint-pd} is equivalent to $||\sqrt{N_y}\sqrt{\gamma_Q}\sqrt{M_x}||_2^2,$ 
    then from the inequality $H(\povm{M})_{\gamma}+H(\povm{N})_{\gamma}\ge H(XY)$ we have:
    \begin{align*}
        &H(\povm{M})_{\gamma}+H(\povm{N})_{\gamma}\\
        &\ge H(XY) \\
        &= -\sum_{x,y}\Pr\{y\leftarrow x\}\log_2\|\sqrt{N_y}\sqrt{\gamma_Q}\sqrt{M_x}\|_2^2 \\
        & \ge -\max_{x,y}\log_2\|\sqrt{N_y}\sqrt{\gamma_Q}\sqrt{M_x}\|_2^2.
    \end{align*}
\end{proof}

A second retrodictive entropic uncertainty relation follows from the observation that, while in Section~\ref{sec:3} we noted that the predictive joint probability distribution $\Pr\{x \to y\}$ generally depends on the choice of instrument for the first measurement and is therefore not symmetric, there exists a \textit{particular} instrument that reproduces the retrodictive joint statistics $\Pr\{y \leftarrow x\}$. To construct this instrument, for each $x$ we choose and fix a unitary extension $U_{\gamma,x}$ of the partial isometry in the polar decomposition of $\sqrt{\gamma_Q}\sqrt{M_x}$. The unitaries $V_{\gamma,y}$ below are chosen and fixed analogously for $\sqrt{\gamma_Q}\sqrt{N_y}$. Explicitly, we define
\begin{align}\label{eq:polar-dec}
\mathcal{J}^{\povm{M}}_x(\gamma_Q) &:= U_{\gamma,x}\sqrt{M_x}\gamma_Q\sqrt{M_x}U_{\gamma,x}^\dag\\
&= \sqrt{\gamma_Q}M_x\sqrt{\gamma_Q}\;,\nonumber
\end{align}
for all $x \in \set{X}$. The unitary operators $U_{\gamma,x}$ depend on the state $\gamma_Q$. Consequently, in general, this construction works for a fixed state, but the map $\gamma_Q\mapsto\sqrt{\gamma_Q}M_x\sqrt{\gamma_Q}$ is not linear in $\gamma_Q$, and no single instrument can reproduce it for all quantum states.

Concatenating this instrument with the square-root instrument for $\povm{N}$ and adding two classical registers for the measurement outcomes, we obtain a quantum channel with hybrid quantum–classical output:
\begin{align}\label{eq:JMN}
    \tilde{\mathcal{J}}^{\povm{M}\to\povm{N}}(\cdot) &:= \sum_{y,x} \sqrt{N_y} U_{\gamma,x} \sqrt{M_x} (\cdot) \sqrt{M_x} U_{\gamma,x}^\dag \sqrt{N_y} \otimes \ketbra{x}_X \otimes \ketbra{y}_Y\;,
\end{align}
which, by construction, when applied to $\gamma_Q$, satisfies
\begin{align}
    \tilde{\mathcal{J}}^{\povm{M}\to\povm{N}}(\gamma_Q) &= \sum_{y,x} \sqrt{N_y}\sqrt{\gamma_Q}M_x\sqrt{\gamma_Q}\sqrt{N_y} \otimes \ketbra{x}_X \otimes \ketbra{y}_Y\;.
\end{align}
Since the joint probability distribution we aim to reproduce is symmetric under the exchange $\povm{M}\leftrightarrow\povm{N}$, we can repeat the same reasoning with the roles of the two POVMs reversed, taking $\povm{N}$ as the first measurement instead of $\povm{M}$. This yields
\begin{align}\label{eq:JNM}
    \tilde{\mathcal{J}}^{\povm{N}\to\povm{M}}(\cdot) &:= \sum_{y,x} \sqrt{M_x} V_{\gamma,y} \sqrt{N_y} (\cdot) \sqrt{N_y} V_{\gamma,y}^\dag \sqrt{M_x} \otimes \ketbra{x}_X \otimes \ketbra{y}_Y\;,
\end{align}
and
\begin{align}
   \tilde{\mathcal{J}}^{\povm{N}\to\povm{M}}(\gamma_Q) &= \sum_{y,x} \sqrt{M_x}\sqrt{\gamma_Q}N_y\sqrt{\gamma_Q}\sqrt{M_x} \otimes \ketbra{x}_X \otimes \ketbra{y}_Y\;. 
\end{align}
The following argument will utilise the definition of the Groenewold--Ozawa information gain that is recalled here:
\begin{definition}[Groenewold--Ozawa information gain~\cite{groenewold1971problem,ozawa1986information}]\label{def:GO-info}
Let $\gamma_Q$ be a state on $Q$ and
\[
  \mathcal{I}(\cdot)
    = \sum_{z} \mathcal{I}_{z}(\cdot)
      \otimes \ketbra{z}_Z
\]
an instrument with classical outcomes $z\in\mathcal{Z}$.
The {\em Groenewold--Ozawa information gain}
of the instrument $\mathcal{I}$ on the input state $\gamma_Q$ is
\begin{equation}\label{eq:GO-info}
  I_{GO}(\gamma_Q;\mathcal{I})
  := H(\gamma_Q)
     - \sum_{z:\Pr(z)>0} \Pr(z)\,
       H\left(\frac{\mathcal{I}_{z}(\gamma_Q)}{\Pr(z)}\right),
\end{equation}
where $\Pr(z)$ is the probability of getting the outcome $z$ and $H(\cdot)$ denotes the von Neumann entropy.
\end{definition}

\

We can now state the following alternative retrodictive entropic uncertainty relation:

\begin{theorem}\label{th:REUR-2}
Let $\gamma_Q$ be a state on $Q$ and let
$\povm{M}=\{M_x:x\in\set{X}\}$ and
$\povm{N}=\{N_y:y\in\set{Y}\}$ be POVMs on $Q$.
Let $\tilde{\mathcal{J}}^{\povm{M}\to\povm{N}}$ and
$\tilde{\mathcal{J}}^{\povm{N}\to\povm{M}}$ be the instruments defined in Eqs.~\eqref{eq:JMN} and~\eqref{eq:JNM} above, and define
\begin{align}\label{eq:gamma-tilda}
  \tilde\gamma_Q
  &:= \bigl[(\tilde{\mathcal{J}}^{\povm{M}\to\povm{N}})^\dag
            \circ\tilde{\mathcal{J}}^{\povm{M}\to\povm{N}}\bigr](\gamma_Q) \\
  &= \sum_{x,y}
     \sqrt{M_x}\, U_{\gamma,x}^\dag N_y
     \sqrt{\gamma_Q} M_x \sqrt{\gamma_Q} N_y
     U_{\gamma,x} \sqrt{M_x}, \nonumber
\end{align}
and
\begin{align}\label{eq:eta-tilda}
  \tilde\eta_Q
  &:= \bigl[(\tilde{\mathcal{J}}^{\povm{N}\to\povm{M}})^\dag
            \circ\tilde{\mathcal{J}}^{\povm{N}\to\povm{M}}\bigr](\gamma_Q) \\
  &= \sum_{x,y}
     \sqrt{N_y}\, V_{\gamma,y}^\dag M_x
     \sqrt{\gamma_Q} N_y \sqrt{\gamma_Q} M_x
     V_{\gamma,y} \sqrt{N_y}\;. \nonumber
\end{align}
Note that, in general, $\Tr\{\tilde\gamma_Q\}\le 1$ and $\Tr\{\tilde\eta_Q\}\le 1$. Then the following retrodictive entropic uncertainty relation holds:
\begin{align}\label{eq:REUR-2}
  &H(\povm{M})_\gamma + H(\povm{N})_\gamma \nonumber\\
  &\ge
  \max\Bigl\{
    D\bigl(\gamma_Q \big\| \tilde\gamma_Q\bigr)
      + I_{GO}\bigl(\gamma_Q;\tilde{\mathcal{J}}^{\povm{M}\to\povm{N}}\bigr)\;,
    \;D\bigl(\gamma_Q \big\| \tilde\eta_Q\bigr)
      + I_{GO}\bigl(\gamma_Q;\tilde{\mathcal{J}}^{\povm{N}\to\povm{M}}\bigr)
  \Bigr\}\;,
\end{align}
where $D(\cdot\|\cdot)$ denotes the Umegaki quantum relative entropy. 

The inequality holds for every fixed choice of these unitary extensions, although its right-hand side may depend on the choice.
\end{theorem}
We report the proof in Appendix~\ref{app:proofIII7}.

\begin{remark}
Two additional observations are worth noting. First, since the channels $\tilde{\mathcal{J}}^{\povm{M}\to\povm{N}}(\cdot)$ and $\tilde{\mathcal{J}}^{\povm{N}\to\povm{M}}(\cdot)$ in Eqs.~\eqref{eq:JMN} and~\eqref{eq:JNM} are both \textit{subunital}, i.e., $\tilde{\mathcal{J}}^{\povm{M}\to\povm{N}}(\openone_Q)\le\openone_{QXY}$ and $\tilde{\mathcal{J}}^{\povm{N}\to\povm{M}}(\openone_Q)\le\openone_{QXY}$, both  operators $\tilde\gamma_Q$ and $\tilde\eta_Q$ are subnormalized, i.e., $\Tr\{\tilde\gamma_Q\}\le 1$ and $\Tr\{\tilde\eta_Q\}\le 1$. Then, by Klein's inequality, this immediately implies that
\[
D(\gamma_Q\|\tilde\gamma_Q)\ge 0
\qquad\text{and}\qquad
D(\gamma_Q\|\tilde\eta_Q)\ge 0\;,
\]
so the relative entropy terms appearing in Eq.\,\eqref{eq:REUR-2} are always 
non-negative.

Second, the instruments $\tilde{\mathcal{J}}^{\povm{M}\to\povm{N}}$ and 
$\tilde{\mathcal{J}}^{\povm{N}\to\povm{M}}$ defined in 
Eqs.\,\eqref{eq:JMN} and \eqref{eq:JNM} are \textit{efficient}~\cite{jacobs2009second}, in the sense that each classical outcome $(x,y)$ is associated with a single Kraus operator. For efficient instruments, the Groenewold--Ozawa information gain is known to be non-negative and, moreover, to coincide with Winter's \textit{intrinsic information} of the corresponding effective POVMs~\cite{Winter01a,buscemi2008global,wilde2012-review-meas-compression}. In the present setting, these POVMs have elements
\begin{align*}
J_{xy}:=\sqrt{M_x}U_{\gamma,x}^\dag N_y U_{\gamma,x}\sqrt{M_x}\;, \qquad\text{and}\qquad J'_{yx}:=\sqrt{N_y}V_{\gamma,y}^\dag M_x V_{\gamma,y}\sqrt{N_y}\;,
\end{align*}
respectively.
\end{remark}

\

\begin{remark}
The Groenewold--Ozawa information gain terms in Eq.~\eqref{eq:REUR-2} appear to be related to a measure of classicality, postulated in Ref.~\cite{Korzekwa2014} and derived in Ref.~\cite{hall-2023-asymmetry}, with which they coincide in the case of a single projective measurement; see, for example, Eq.~(23) in~\cite{hall-2023-asymmetry}. In the present setting, however, the Groenewold--Ozawa information gain naturally arises also for general POVMs, generalizing its operational use beyond the projective setting. A thorough examination of this connection is left for future work.
\end{remark}

The bound in Theorem~\ref{th:REUR-2} admits a particularly transparent form when at least one element of the measurement–state triple is rank-one, as summarized in the following corollary.

\begin{corollary}\label{rem:GO-info-bound}
\sloppy Assume now at least one among $\gamma_Q$, $\povm{M}$, or $\povm{N}$ is rank-one. In this case~\cite{buscemi2008global}
\[
I_{GO}(\gamma_Q;\tilde{\mathcal{J}}^{\povm{M}\to\povm{N}}) = I_{GO}(\gamma_Q;\tilde{\mathcal{J}}^{\povm{N}\to\povm{M}})=H(\gamma_Q)\;,
\]
and hence,
    \begin{align}\label{eq:when-rank-one}
            &H(\povm{M})_\gamma + H(\povm{N})_\gamma \nonumber\\
            &\geq H(\gamma_Q)+\max\left\{D\bigl(\gamma_Q \| \tilde\gamma_Q\bigr)\;,\;D\bigl(\gamma_Q \| \tilde\eta_Q\bigr)\right\}\;.
    \end{align}
\end{corollary}

\subsection{Numerical Benchmark Codes and Results}\label{sec:numericalresult}

In this section we present the benchmark code together with the corresponding numerical results. 

We compare our bounds~\eqref{eq:REUR-1} and~\eqref{eq:REUR-2} with the well-known entropic uncertainty relation of Berta \textit{et al}.~\cite{berta2010uncertainty}:
\begin{equation}\label{eq:bertaEUR}
    H(\povm{M})_\gamma +H(\povm{N})_\gamma\ge -\max_{x,y} \log_2 \left\|\sqrt{N_y}\sqrt{M_x}\right\|^{2}_{\infty}+H(\gamma_Q),
\end{equation}
where the measurements $\{M_x\}_x$ and $\{N_y\}_y$ are rank-one projective valued measurements (PVMs). Our bounds~\eqref{eq:REUR-1} and~\eqref{eq:REUR-2} apply to general POVMs; however, for a fair comparison with~\eqref{eq:bertaEUR}, we restrict to the rank-one PVM setting in this subsection.

For general POVMs, Eq.~\eqref{eq:bertaEUR} does not hold, although it does apply whenever at least one of the POVMs has rank-one elements~\cite{coles2011information}. Motivated by this fact, we also perform numerical tests on more general rank-one POVMs.

The three bounds compared in the numerical experiments are:
\begin{align*}
    \mathrm{EUR1} &:= -\max_{x,y}\log_2\left\|\sqrt{N_y}\sqrt{\gamma_Q}\sqrt{M_x}\right\|_2^2, \\
    \mathrm{EUR2} &:= H(\gamma_Q) + \max\left\{D\bigl(\gamma_Q \| \tilde\gamma_Q\bigr),D\bigl(\gamma_Q \| \tilde\eta_Q\bigr)\right\},\\
    \mathrm{EUR3} &:= H(\gamma_Q) -\max_{x,y}\log_2\left\|\sqrt{N_y}\sqrt{M_x}\right\|^{2}_{\infty}.
\end{align*}
We perform $100{,}000$ random trials for each configuration. The results are summarized in Table~\ref{table:Gaps_pvm} and Figure~\ref{fig:benchmark_pvm} for rank-one PVMs, and in Table~\ref{table:Gaps_povm} and Figure~\ref{fig:benchmark_povm} for general rank-one POVMs. Here $d$ denotes the Hilbert space dimension and $n$ the number of measurement outcomes.

\begin{table*}[t]
\centering
\setlength{\tabcolsep}{10pt}
\renewcommand{\arraystretch}{1.25}

\begin{tabular}{lccc}
\hline\hline
\textbf{Number of counts} 
& \textbf{$d=2,\; n=2$} 
& \textbf{$d=3,\; n=3$} 
& \textbf{$d=4,\; n=4$} \\
\hline
EUR1 $<$ EUR3 & 76325 & 44045 & 39426 \\
EUR2 $<$ EUR3 & 394   & 0     & 0     \\
EUR2 $<$ EUR1 & 0     & 0     & 0     \\
\hline\hline
\end{tabular}

\caption{Counts of random trials (out of $100{,}000$) in which one uncertainty bound is strictly weaker than another for rank-one PVMs in dimension $d$ with $n=d$ outcomes. ``Number of counts'' for, e.g., ``EUR1 $<$ EUR3'', means the number of cases in which the bound~\eqref{eq:REUR-1} is numerically smaller (i.e., weaker) than the Berta \textit{et al.} bound~\eqref{eq:bertaEUR}.}
\label{table:Gaps_pvm}
\end{table*}

\begin{table*}[t]
\centering
\setlength{\tabcolsep}{10pt}
\renewcommand{\arraystretch}{1.25}

\begin{tabular}{lccc}
\hline\hline
\textbf{Number of counts}
& \textbf{$d=2,n=3$}
& \textbf{$d=2,n=4$}
& \textbf{$d=3,n=4$} \\
\hline
EUR1 $<$ EUR3 &  46123 & 44905  & 41855 \\
EUR2 $<$ EUR3 &  0     & 0          & 0    \\
EUR2 $<$ EUR1 &  0     & 0          & 0      \\
\hline\hline
\end{tabular}
\caption{Counts of negative gaps (out of $100{,}000$ random trials) for rank-one POVMs with varying dimensions $d$ and outcomes $n>d$.}
\label{table:Gaps_povm}
\end{table*}

%%%%%%%%%%%%%%%%%%%%%%%%%%%%%%%%%%%%%

In Table~\ref{table:Gaps_pvm} we generate random rank-one PVMs. Since any rank-one PVM in dimension $d$ must have exactly $d$ outcomes, we only consider the case $n=d$. The numerical results show that the bound~\eqref{eq:REUR-1} is not uniformly comparable to the bound of Berta et al., while bound~\eqref{eq:REUR-2} is consistently stronger than~\eqref{eq:bertaEUR} for all cases with $n\ge 3$, and also always stronger than~\eqref{eq:REUR-1}.

We emphasize, however, that in general bound~\eqref{eq:REUR-2} is not strictly comparable with the Berta \textit{et al.} bound~\eqref{eq:bertaEUR}. Although Table~\ref{table:Gaps_pvm} shows no negative gaps for $d \ge 3$ in the random tests, explicitly fine-tuned counterexamples can be constructed. Consider mutually unbiased PVMs $\povm{M}$ and $\povm{N}$ and a state of the form
\[
\gamma_Q = (1-p)\frac{\openone}{d} + p\ketbra{\psi},
\]
where
\[
\ket{\psi} \propto \cos\theta\ket{m_x} + \sin\theta\ket{n_x},
\]
with $\ket{m_x}$ and $\ket{n_y}$ denoting mutually unbiased measurement bases. For parameters near $p\approx 0.75$ and $\theta\approx 45^{\circ}$, negative gaps reproducibly appear, as shown in Figure~\ref{fig:MUB} for $d=3$ and $d=5$.

\begin{figure}[tb]
    \centering
    \includegraphics[width=0.46\textwidth]{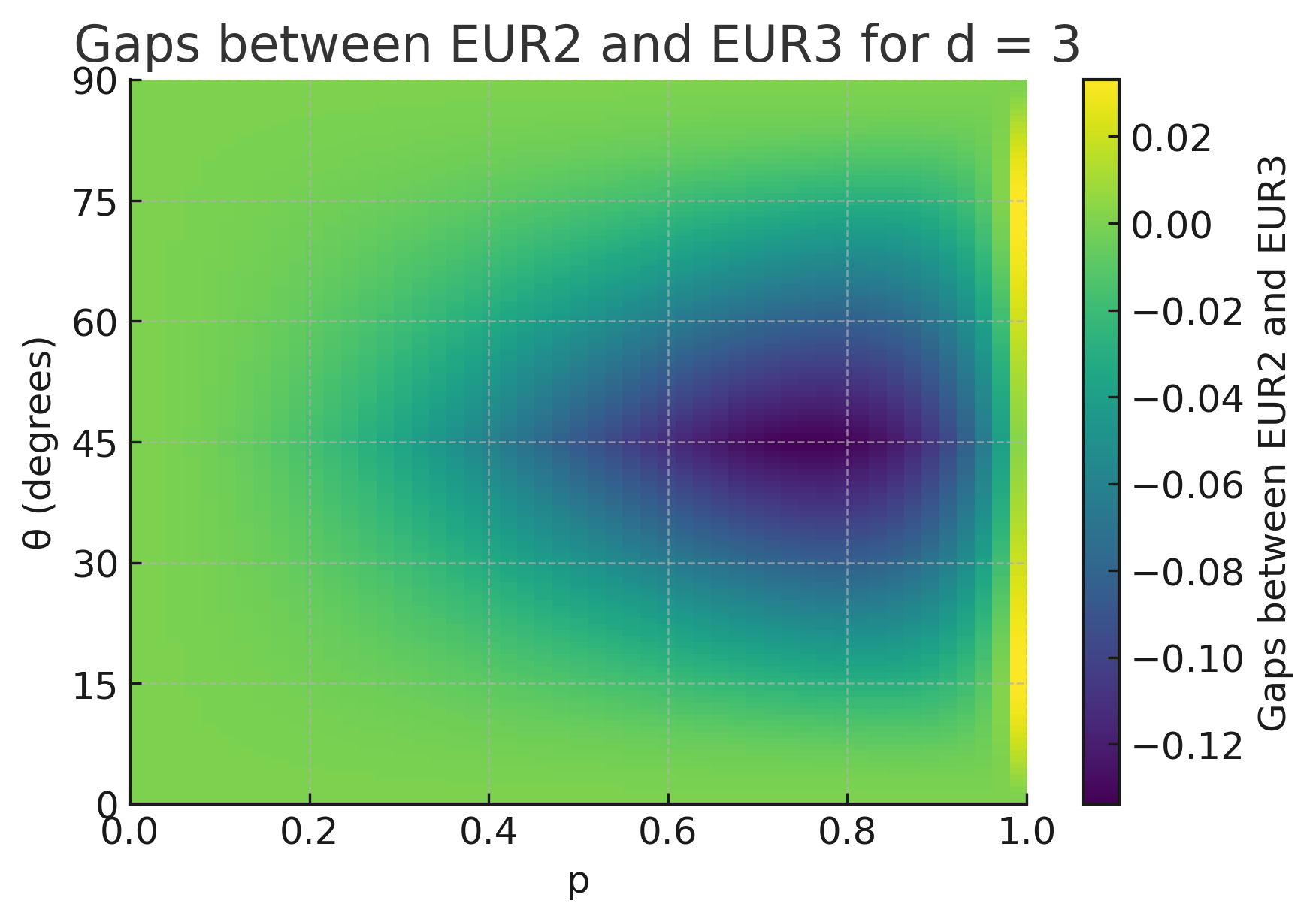}
    \includegraphics[width=0.46\textwidth]{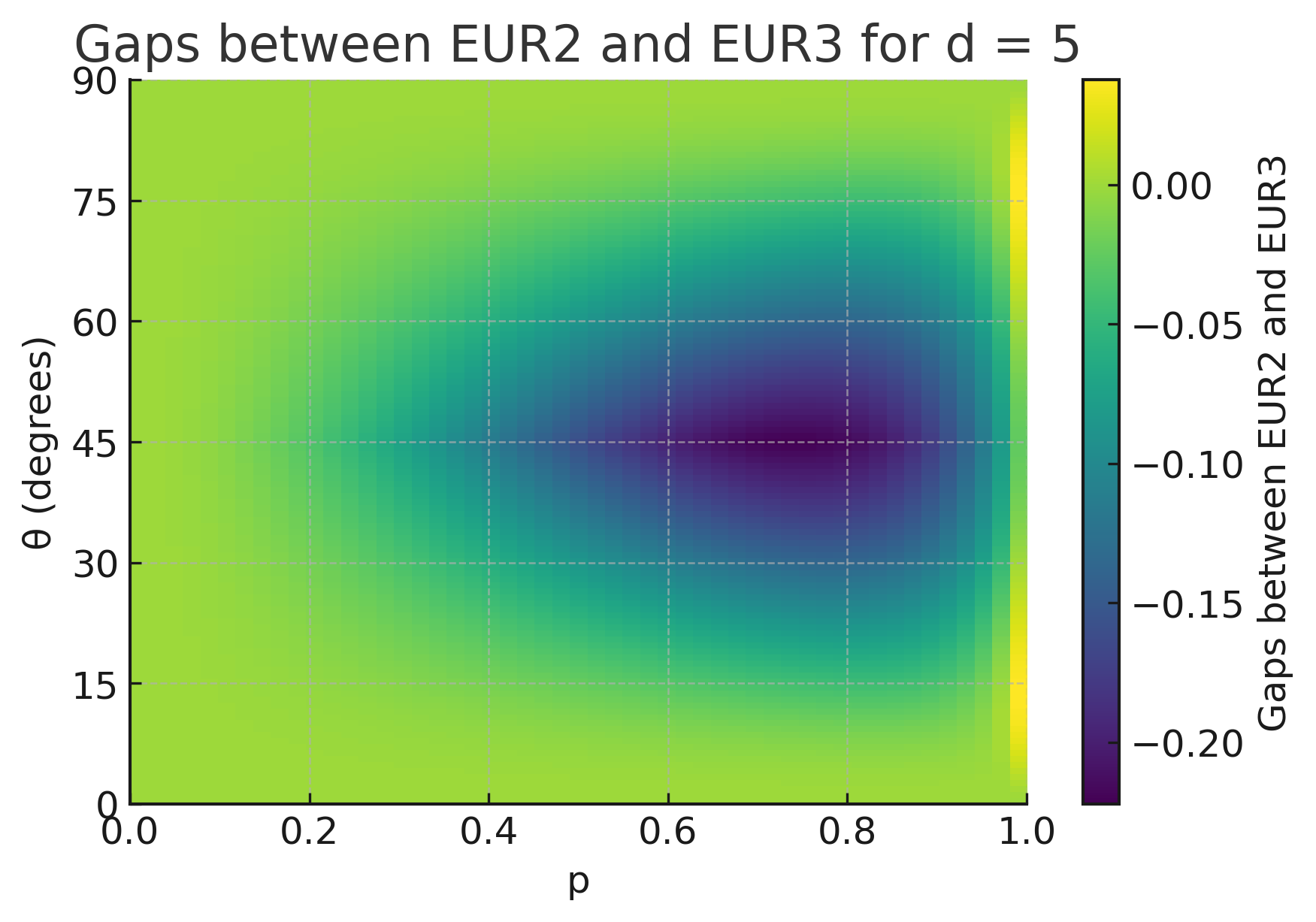}
    \caption{Negative gaps between EUR2~\eqref{eq:REUR-2} and EUR3~\eqref{eq:bertaEUR} for mutually unbiased PVMs in dimensions $d=3$ and $d=5$. The shaded regions correspond to parameter ranges $(p,\theta)$ where EUR2 becomes weaker than EUR3. These counterexamples occur for states that interpolate between the two bases and illustrate that EUR2 is not uniformly stronger than the Berta \textit{et al.} bound.}
    \label{fig:MUB}
\end{figure}

These figures indicate that such negative gaps persist over a nontrivial region of $(p,\theta)$ that depends on the dimension. Exact mutual unbiasedness is not required; near-MUB measurements exhibit the same behaviour. Random generation rarely produces such pairs in higher dimensions, which explains their absence from the general PVM tests. These counterexamples highlight a key conceptual difference: our retrodictive entropic uncertainty relation depends on the joint interplay between the state and the two measurements, whereas the Berta et al. relation does not incorporate this dependence.

As we extend to more general rank-one POVMs, similar behaviour appears (see Table~\ref{table:Gaps_povm} and Figure~\ref{fig:benchmark_povm}). The upper limit of EUR3 is now attained for uniformly split mutually unbiased POVMs of the form $\frac{1}{m}\ketbra{m_x}$ and $\frac{1}{m}\ketbra{n_y}$ with $m$ being some constant.

\section{Conclusion}\label{sec:conclusion}

In this paper, we start from the quantum minimum change principle and show that, for all quantum-to-classical channels, the information projection \eqref{eq:q_optimize} onto fixed outcome statistics selects a single retrodictive update independent of the divergence used. This yields a canonical and divergence-independent quantum Bayesian inverse \eqref{eq:optimal-retrodiction} for any POVM and prior state, extending and unifying earlier constructions. Using this update, we introduced a symmetric retrodictive joint distribution \eqref{eq:retro-joint-pd} for pairs of POVMs and defined the mutual retrodictability $R(\povm{M};\povm{N})_\gamma$. Then, Theorem~\ref{th:mutual-bound-2} shows that this quantity is always finite and bounded above by $H(\gamma_Q)$.

We then used the retrodictive joint distribution to derive two entropic uncertainty relations, Eqs.~\eqref{eq:REUR-1} and~\eqref{eq:REUR-2}, which express uncertainty directly in terms of the prior state and the retrodictive action of the POVMs. The second bound links the entropic trade-off to relative entropies between the prior and its retrodictively updated versions, giving the relation a clear information-theoretic interpretation. Furthermore, the second bound is an always-valid form expressed via the Groenewold--Ozawa information gain, making explicit the connection between our entropic uncertainty relations and information-theoretic approaches to the thermodynamics of quantum measurements~\cite{Minagawa2025SecondLaw}.

Finally, we benchmarked our retrodictive uncertainty relations against the Berta \textit{et al.} bound~\eqref{eq:bertaEUR}, both for rank-one PVMs and for rank-one POVMs in low dimensions. The numerical evidence shows that the bound EUR2~\eqref{eq:REUR-2} is consistently tighter across broad classes of measurements and states, while also identifying structured counterexamples such as near mutually unbiased settings.

\
\section*{Data availability statement}
The data that support the findings of this study are openly available at the following \url{https://doi.org/10.5281/zenodo.17636821}.

\section*{Acknowledgements} 
We warmly thank Michael J.W. Hall for his valuable comments, suggestions, and encouragement. J.~K. and K.~T. acknowledge support from JST SPRING Grant No. JPMJSP2125. All authors acknowledge support from MEXT Quantum Leap Flagship Program (MEXT QLEAP) Grant No. JPMXS0120319794; from MEXT-JSPS  Grant-in-Aid for Transformative Research Areas  (A) ``Extreme Universe,'' No.~21H05183; and  from JSPS  KAKENHI Grants No.~23K03230 and No.~26K00621.

\bibliographystyle{apsrev4-2}
\bibliography{library}

\appendix
\newpage

\section{Derivation of the bipartite state \texorpdfstring{$\gamma_{XQ}$}{gamma(XQ)} in Eq.~\texorpdfstring{\eqref{eq:forward-cq-state}}{(reference)}}\label{app:gammaXQ}

Following the definition of ``forward'' output-input bipartite state $\mathcal{Q}_{\text{fwd}}$ in Ref.~\cite{bai-2025-q-bayes-minimum-change}, consider the canonical purification of $\gamma_Q$:
\[
    \dket{\sqrt{\gamma_Q}} := \sum_{ij}\bra{i}\sqrt{\gamma_Q}\ket{j}\ \ket{i}_{Q_1}\ket{j}_{Q_2}\, ,
\]
where $Q_1\cong Q_2\cong Q$ and $\{\ket{i}\}$ is the chosen orthonormal basis of $\Hil_{Q}$. For the quantum-classical measurement channel defined in~\eqref{eq:qc-channel}, we obtain
\begin{align*}
    \mathcal{Q}_{\text{fwd}}^{\mathcal{M}} &= (\mathcal{M}\otimes\mathcal{I})(\dket{\sqrt{\gamma_Q}}\dbra{\sqrt{\gamma_Q}})\\
    &= (\openone_X \otimes \sqrt{\gamma_{Q}^{T}})\ C_{\mathcal{M}}\ (\openone_X \otimes \sqrt{\gamma_{Q}^{T}})\\
    &= (\openone_X \otimes \sqrt{\gamma_{Q}^{T}})\ (\sum_{i,j}\mathcal{M}(\ket{i}\!\bra{j})\otimes\ket{i}\!\bra{j})\ (\openone_X \otimes \sqrt{\gamma_{Q}^{T}})\\
    &=(\openone_X \otimes \sqrt{\gamma_{Q}^{T}})\ (\sum_{i,j}\sum_{x\in\mathcal{X}}\bra{j}M_x \ket{i}\ketbra{x} \otimes \ket{i}\!\bra{j})\ (\openone_X \otimes \sqrt{\gamma_{Q}^{T}})\\
    &= \sum_{x\in\mathcal{X}}\ketbra{x} \otimes \sqrt{\gamma_{Q}^{T}}\big(\sum_{i,j}\bra{j}M_x \ket{i}\ket{i}\!\bra{j}\big)\sqrt{\gamma_{Q}^{T}}\\
    &= \sum_{x\in\mathcal{X}}\ketbra{x} \otimes \sqrt{\gamma_{Q}^{T}} M_x^{T} \sqrt{\gamma_{Q}^{T}}\;,
\end{align*}
where $C_{\mathcal{M}}=\sum_{i,j}\mathcal{M}(\ket{i}\!\bra{j})\otimes\ket{i}\!\bra{j}$ is the Choi operator of the measurement channel. Finally, since in this case the operator $\mathcal{Q}_{\text{fwd}}^{\mathcal{M}}$ is separable, we suppress the partial transposition on the second space for the sake of simplicity and readability, which leads to our definition of $\gamma_{XQ}$ in Eq.~\eqref{eq:forward-cq-state}.

\section{Proof of Theorem~\ref{th:divergences} and Theorem~\ref{th:trace-dist}}\label{app:proofs-secII}

\begin{proof}[Proof of Theorem~\ref{th:divergences}]
    The claim follows from three properties satisfied by all divergences listed above: the direct-sum property, monotonicity under partial trace, and faithfulness \cite{khatri2024principlesquantumcommunicationtheory}. In particular, monotonicity implies that
\[
    \min_{\sigma_{XQ}}\mathbb{D}(\sigma_{XQ}\|\gamma_{XQ})
    \ge \mathbb{D}(\tau\|p)\;,
\]
so that any state that saturates the inequality is necessarily a global minimizer.

The argument proceeds in the same way for each divergence in the list. For illustration, we discuss two representative cases. First, consider the Umegaki relative entropy, for which we have
\begin{align*}
        &D(\sigma_{XQ}\|\gamma_{XQ})\\
        &=D\Big( \sum_{x\in\set{X}}\tau(x)\ketbra{x}_X\otimes\sigma_{Q|x} \Big\|  \sum_{x\in\set{X}} p(x) \ketbra{x}_X\otimes\frac{\sqrt{\gamma_Q}M_{x}\sqrt{\gamma_Q}}{\Tr\{\gamma_Q M_{x}\}} \Big)\\
        &= D(\tau\|p) + \sum_{x\in\set{X}}\tau(x)D\Big(\sigma_{Q|x}\Big\|\frac{\sqrt{\gamma_Q}M_{x}\sqrt{\gamma_Q}}{\Tr\{\gamma_Q M_{x}\}}\Big)\;, 
    \end{align*}
    where the second equality is verified by the direct-sum property.
    The minimum of $D(\sigma_{XQ}\|\gamma_{XQ})$ is achieved when the second term vanishes, and from the faithfulness this means
    \begin{align}
    \sigma_{Q|x}=\frac{\sqrt{\gamma_Q}M_{x}\sqrt{\gamma_Q}}{\Tr\{\gamma_Q M_{x}\}} \label{sigma_Qx}
    \end{align}
    for all $x\in\set{X}$ with $\tau(x)>0$. 
    Next, for the sandwiched R\'enyi relative entropy, we obtain
    \begin{align*}
         &\tilde{D}_\alpha(\sigma_{XQ}\|\gamma_{XQ})\\
         &= \tilde{D}_\alpha\Big( \sum_{x\in\set{X}}\tau(x)\ketbra{x}_X\otimes\sigma_{Q|x} \Big\|\sum_{x\in\set{X}} p(x) \ketbra{x}_X\otimes\frac{\sqrt{\gamma_Q}M_{x}\sqrt{\gamma_Q}}{\Tr\{\gamma_Q M_{x}\}} \Big)\\
        &= \frac{1}{\alpha-1} \log_2 \sum_{x\in\set{X}} \tau(x)^\alpha p(x)^{1-\alpha}\Tr\left\{\left( \Big(\frac{\sqrt{\gamma_Q}M_x\sqrt{\gamma_Q}}{\Tr\{\gamma_Q M_x\}}\Big)^{\frac{1-\alpha}{2\alpha}} \sigma_{Q|x}\Big(\frac{\sqrt{\gamma_Q}M_x\sqrt{\gamma_Q}}{\Tr\{\gamma_Q M_x\}}\Big)^{\frac{1-\alpha}{2\alpha}} \right)^\alpha\right\}\\
        &=  \frac{1}{\alpha-1} \log_2 \sum_{x\in\set{X}} \tau(x)^\alpha p(x)^{1-\alpha} \; 2^E,\\
    \end{align*}
    where
    \[
    E:=(\alpha-1)\tilde{D}_\alpha \Big(\sigma_{Q|x}\Big\|\frac{\sqrt{\gamma_Q}M_x\sqrt{\gamma_Q}}{\Tr\{\gamma_Q M_x\}}\Big)\;.
    \]
    The second equality again follows from the direct-sum property. 
    The minimum of $\tilde{D}_\alpha(\sigma_{XQ}\|\gamma_{XQ})$ is achieved when 
    \[
    \tilde{D}_\alpha \Big(\sigma_{Q|x}\Big\|\frac{\sqrt{\gamma_Q}M_x\sqrt{\gamma_Q}}{\Tr\{\gamma_Q M_x\}}\Big)=0
    \]
    for all $x\in\set{X}$ with $\tau(x)>0$ and for all $\alpha \in[1/2,1)\cup(1,\infty)$, thus implying Eq.~\eqref{sigma_Qx}.
    Similar arguments hold for the other cases. 
    Therefore, we conclude
    \begin{align*}
        \min_{\sigma_{XQ}}\mathbb{D}(\sigma_{XQ}\|\gamma_{XQ})=\mathbb{D}(\sigma_X\|\gamma_X),
    \end{align*}
    and the minimum is achieved by $\sigma_{XQ}=\hat{\gamma}_{XQ}$.\\
\end{proof}

\

\begin{proof}[Proof of Theorem~\ref{th:trace-dist}]
    Due to the common cq structure, we have
    \begin{align*}
    \left\|\sigma_{XQ}-\gamma_{XQ}\right\|_1&=\sum_{x\in\set{X}}\left\|\tau(x)\sigma_{Q|x}-p(x)\frac{\sqrt{\gamma_Q}M_x\sqrt{\gamma_Q}}{\Tr\{\gamma_Q M_x\}}\right\|_1\\
    &\ge \sum_{x\in\set{X}}\left|\tau(x)-p(x)\right|\;,
 \end{align*}
 where the equality in the second line holds if 
 \[
 \sigma_{Q|x}=\frac{\sqrt{\gamma_Q}M_x\sqrt{\gamma_Q}}{\Tr\{\gamma_Q M_x\}}=:\hat\gamma_{Q|x}\;,
\] 
for all $x$.\\
\end{proof}

\section{Proof of Theorem~\ref{th:mutual-bound-2}}\label{app:proofIII2}

\begin{proof}[Proof of Theorem~\ref{th:mutual-bound-2}]
The mutual retrodictability is defined as the mutual information of the joint distribution
\[
\Pr\{y \leftarrow x\}=\Tr\{N_y\sqrt{\gamma_Q}M_x\sqrt{\gamma_Q}\}\;,
\]
so non-negativity follows immediately. It remains to prove the upper bound.

Choose an orthonormal basis that diagonalizes $\gamma_Q$, writing $\gamma_Q=\sum_i g_i\ketbra{\psi_i}$. Using this basis, define the canonical purification
\[
\ket{\Psi_\gamma}:=\sum_i\sqrt{g_i}\ket{\psi_i}\otimes\ket{\psi_i}\;.
\]
By construction,
\[
\Tr_1\{\ketbra{\Psi_\gamma}\}=\Tr_2\{\ketbra{\Psi_\gamma}\}=\gamma_Q\;.
\]
Moreover, for all operators $X$ and $Y$,
\[
\Tr\{X\otimes Y\ \ketbra{\Psi_\gamma}\}
=\Tr\!\left\{X\sqrt{\gamma_Q}\,Y^T\sqrt{\gamma_Q}\right\},
\]
where the transpose is taken in the eigenbasis of $\gamma_Q$. Hence,
\[
\Pr\{y\leftarrow x\}
=\Tr\{N_y\otimes M_x^T\ \ketbra{\Psi_\gamma}\}\;.
\]
Let $\mathcal{N}$ and $\mathcal{M}^T$ denote the qc channels, as defined in Eq.~\eqref{eq:qc-channel}, associated with $\povm{N}$ and $\povm{M}^T:=\{M_x^T:x\in\set{X}\}$, respectively. The above implies
\[
R(\povm{M};\povm{N})_\gamma
=I(X;Y)_{(\mathcal{N}\otimes\mathcal{M}^T)(\ketbra{\Psi_\gamma})}\;.
\]

We now apply the data–processing inequality for the quantum mutual information:
\begin{align*}
R(\povm{M};\povm{N})_\gamma
&= I(Y;X)_{(\mathcal{N}\otimes\mathcal{M}^T)(\ketbra{\Psi_\gamma})}\\
&\le I(Y;Q)_{(\mathcal{N}\otimes\mathcal{I})(\ketbra{\Psi_\gamma})}\;.
\end{align*}
Since $(\mathcal{N}\otimes\mathcal{I})(\ketbra{\Psi_\gamma})$ is a cq state with marginal $\gamma_Q$, the quantum mutual information satisfies
\[
I(Y;Q)_{(\mathcal{N}\otimes\mathcal{I})(\ketbra{\Psi_\gamma})}
\le H(\gamma_Q)\;.
\]
Combining inequalities yields the claimed bound
\[
R(\povm{M};\povm{N})_\gamma \le H(\gamma_Q)\;.
\]
\end{proof}

\section{Proof of Theorem~\ref{th:REUR-2}}\label{app:proofIII7}

\begin{proof}[Proof of Theorem~\ref{th:REUR-2}]
    Since the same exact reasoning can be repeated substituting $\tilde{\mathcal{J}}^{\povm{M}\to\povm{N}}$ with $\tilde{\mathcal{J}}^{\povm{N}\to\povm{M}}$, we prove the lower bound~\eqref{eq:REUR-2} only for $\tilde\gamma_Q$.
    
    The proof follows the idea in Ref.~\cite{buscemi-das-wilde-2016approximate}. For a state $\gamma_Q$ and $\E:\mathcal{L}(\mathcal{H})\rightarrow \mathcal{L}(\mathcal{H}')$ a positive, trace-preserving and sub-unital map, we have:
    \begin{align*}
        &H(\E(\gamma_Q))-H(\gamma_Q) \\
        &= \Tr\{\gamma_Q \log_2 \gamma_Q\} - \Tr\{\E(\gamma_Q)\log_2\E(\gamma_Q)\} \\
        &= \Tr\{\gamma_Q \log_2 \gamma_Q\} - \Tr\{\gamma_Q \E^{\dag}(\log_2\E(\gamma_Q))\} \\
        &\ge \Tr\{\gamma_Q \log_2 \gamma_Q\} - \Tr\{\gamma_Q \log_2(\E^{\dag}\circ\E)(\gamma_Q)\} \\
        &= D(\gamma_Q\|(\E^{\dag}\circ\E)(\gamma_Q)),\numberthis
    \end{align*}
    where the second equality is from the definition of adjoint and the inequality follows from operator concavity of the logarithm and the operator Jensen inequality for positive  unital maps~\cite{choi1974schwarz}. The sub-unitality is required for the non-negativity of $D(\gamma_Q\|(\E^{\dag}\circ\E)(\gamma_Q)).$ 
    
    \sloppy Noting that the sequential measurement channel $\tilde{\mathcal{J}}^{\povm{M}\to\povm{N}}$ is CPTP, sub-unital, and 
    \begin{align*}
      \tilde{\mathcal{J}}^{\povm{M}\to\povm{N}}(\gamma_Q) &= \sum_{y,x} \Pr\{y\leftarrow x\}\frac{\sqrt{N_y}\sqrt{\gamma_Q}M_x\sqrt{\gamma_Q}\sqrt{N_y}}{\Pr\{y\leftarrow x\}}\otimes \ketbra{x}_X \otimes \ketbra{y}_Y\\
      &=:\omega_{QXY}\;,
    \end{align*}
    we thus have
    \begin{align*}
       &H(\tilde{\mathcal{J}}^{\povm{M}\to\povm{N}}(\gamma_Q)) - H(\gamma_Q) \\
       &= H(QXY)_\omega - H(\gamma_Q) \\
       &= H(XY)_{\omega}\!+\! \sum_{x,y}\Pr\{y\leftarrow x\}H\!\left(\frac{\tilde{\mathcal{J}}^{\povm{M}\to\povm{N}}_{x,y}(\gamma_Q)}{\Pr\{y\leftarrow x\}}\right) \!-\! H(\gamma_Q) \\
       &= H(XY)_{\omega} - I_{GO}(\gamma_Q;\tilde{\mathcal{J}}^{\povm{M}\to\povm{N}}) \\
       &\ge D\left(\gamma_Q\left\|\left[(\tilde{\mathcal{J}}^{\povm{M}\to\povm{N}})^\dag\circ\tilde{\mathcal{J}}^{\povm{M}\to\povm{N}}\right](\gamma_Q)\right.\right) \\
       & = D(\gamma_Q\|\tilde{\gamma}_Q)\;,
    \end{align*}
    where
    \[
    \tilde{\mathcal{J}}^{\povm{M}\to\povm{N}}_{x,y}(\gamma_Q) := \sqrt{N_y}\sqrt{\gamma_Q}M_x\sqrt{\gamma_Q}\sqrt{N_y}\;.
    \]
    
    Note that $(\tilde{\mathcal{J}}^{\povm{M}\to\povm{N}})^\dag$ is trace-non-increasing, so $\Tr\{\tilde\gamma_Q\}\le 1$. Since for Umegaki relative entropy, we have $B \leq B' \Rightarrow D(A||B) \geq D(A||B')$, therefore the use of $\tilde{\gamma}_Q$ actually makes the bound tighter.

    Combine the above equation with the well-known bound:
    \[
    H(\povm{M})_\gamma + H(\povm{N})_\gamma \ge H(XY)_\omega \;,
    \]
    and consider the measurement order swapping, we recover the bound~\eqref{eq:REUR-2}.\\
\end{proof}

\newpage
\begin{sidewaysfigure}[tb]
    \centering
    \includegraphics[width=\textheight]{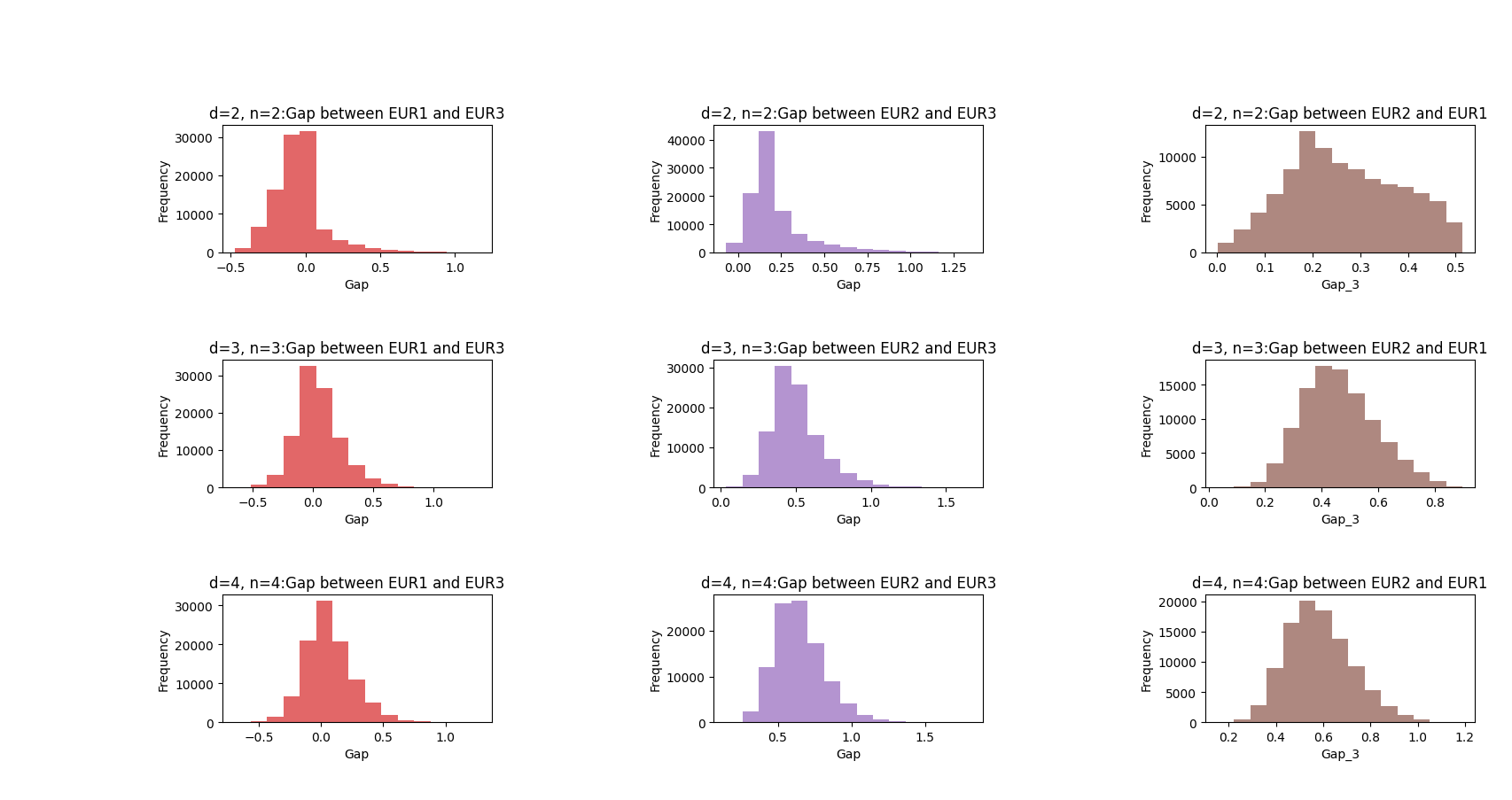}
    \caption{Numerical comparison of the bounds EUR1~\eqref{eq:REUR-1}, EUR2~\eqref{eq:REUR-2}, and EUR3~\eqref{eq:bertaEUR} for random rank-one PVMs. Each point corresponds to one of $100{,}000$ random trials. The plot illustrates that EUR2 is consistently tighter than the Berta \textit{et al.} bound for $d\ge 3$, while EUR1 shows no uniform ordering relative to EUR3.}
    \label{fig:benchmark_pvm}
\end{sidewaysfigure}
\begin{sidewaysfigure}[tb]
    \centering
    \includegraphics[width=\textheight]{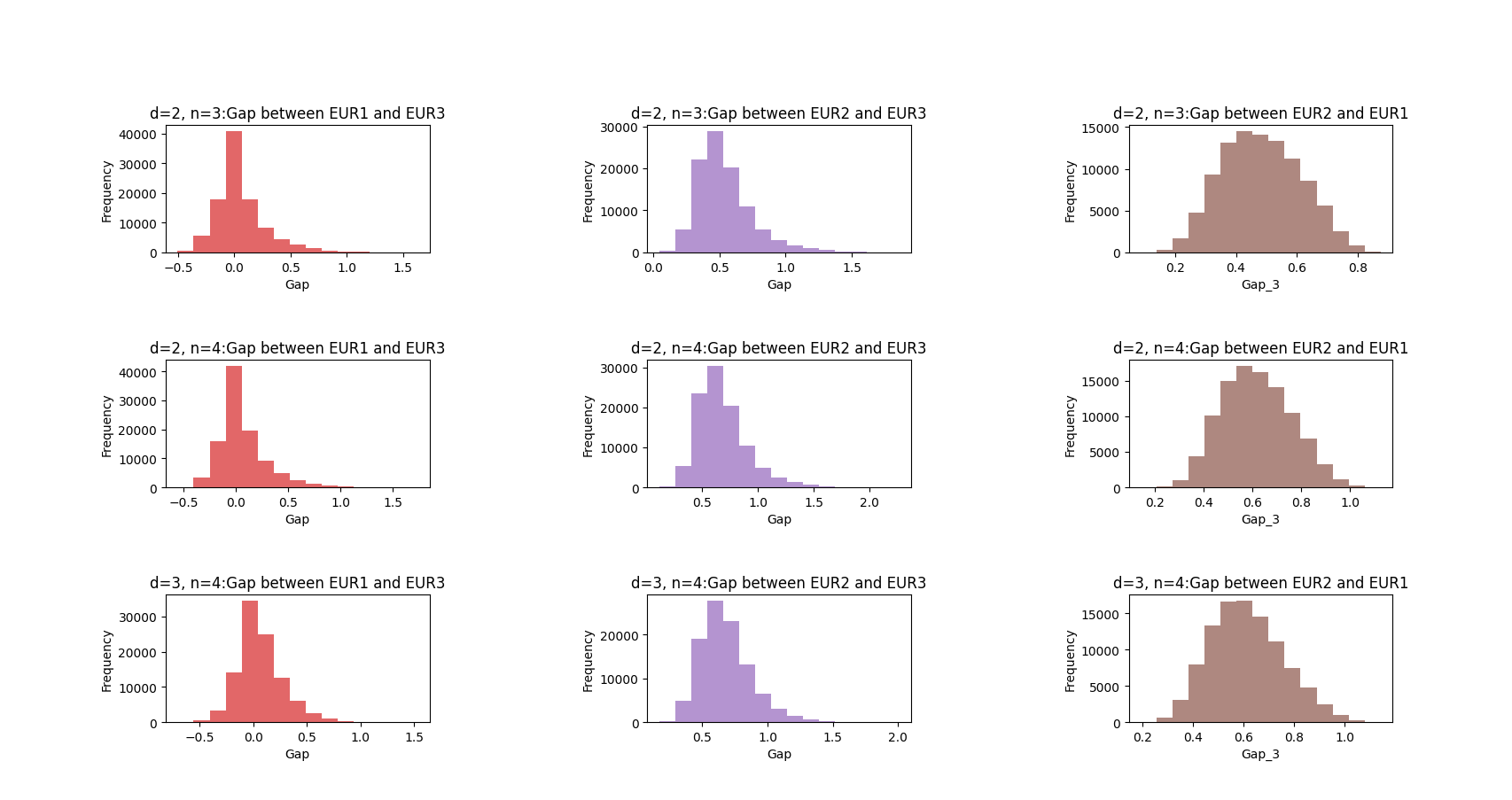}
    \caption{Numerical comparison of EUR1, EUR2, and EUR3 for random rank-one POVMs in various dimensions. The data show the same trend as in the projective case: EUR2 dominates EUR3 in nearly all instances, while EUR1 exhibits mixed behaviour. The increased variety of POVMs broadens the range of achievable overlaps and highlights the structural advantage of the retrodictive bound.}
    \label{fig:benchmark_povm}
\end{sidewaysfigure}

\end{document}